\documentclass[letterpaper,11pt]{article}
\usepackage[letterpaper,margin=1in]{geometry}
\usepackage{amsfonts}
\usepackage{amsmath}
\usepackage{amsthm}
\usepackage{amssymb}
\usepackage{algorithm}
\usepackage[noend]{algpseudocode}
\usepackage[hyphens]{url}
\usepackage{hyperref}
\usepackage{graphicx}
\usepackage{subfigure}
\usepackage{xspace}
\usepackage{units}

\usepackage{paralist}

\newtheorem*{theorem*}{Theorem}
\newtheorem{theorem}{Theorem}
\newtheorem{corollary}[theorem]{Corollary}
\newtheorem{proposition}[theorem]{Proposition}
\newtheorem{lemma}[theorem]{Lemma}

\newtheorem{conjecture}[theorem]{Conjecture}
\newtheorem{property}[theorem]{Property}

\newcommand{\OP}[2]{\ensuremath{\overrightarrow{\mathcal P}_{\!#1}^{#2}}\xspace}
\newcommand{\UP}[2]{\ensuremath{\mathcal P_{\;\!#1}^{\;\!#2}}\xspace}
\newcommand{\OR}[2]{\ensuremath{\overrightarrow{\mathcal R}_{#1}^{#2}}\xspace}
\newcommand{\UR}[2]{\ensuremath{\mathcal R_{#1}^{#2}}\xspace}

\newcommand{\FS}[1]{\ensuremath{{\mathcal F}(#1)}\xspace}

\begin{document}

\title{Universal  Systems of Oblivious Mobile Robots\thanks{This work has been supported in part by 
 the Natural Sciences and Engineering Research Council of Canada through the Discovery Grant program; by Prof. Flocchini's University Research Chair; and by the Scientific Grant in Aid by the Ministry of  Education, Culture, Sports, Science, and Technology of Japan.} }
\author{Paola Flocchini$^\dagger$, Nicola Santoro$^\ddagger$, Giovanni Viglietta$^\dagger$, Masafumi Yamashita$^\S$\medskip\\
  \small
  $^\dagger$ University of Ottawa, Canada. {\tt  paola.flocchini@uottawa.ca}, {\tt gvigliet@uottawa.ca} \\
  \small
  $^\ddagger$ Carleton University, Ottawa, Canada.  {\tt santoro@scs.carleton.ca}\\
  \small
  $^\S$ Kyushu University, Fukuoka, Japan. {\tt  mak@csce.kyushu\nobreakdash-u.ac.jp}
} 
\date{}

\index{Flocchini, Paola}
\index{Santoro, Nicola}
\index{Viglietta, Giovanni}
\index{Yamashita, Masafumi}

\maketitle
\thispagestyle{empty}


\begin{abstract}
\noindent An oblivious mobile robot is a stateless computational entity located in a spatial universe, capable of moving in that universe. When activated, the robot observes the universe and the location of the other robots, chooses a destination, and moves there. The computation of the destination is made by executing an algorithm, the same for all robots, whose sole input is the current observation.
No memory of all these actions is retained after the move. When the spatial universe is a graph, distributed computations by oblivious mobile robots have been intensively studied focusing on the conditions for feasibility of basic problems (e.g., gathering, exploration) in specific classes of graphs under different schedulers. In this paper, we embark on a different, more general, type of investigation.

With their movements from vertices to neighboring vertices, the robots make the system transition from one configuration to another. Thus the execution of an algorithm from a given configuration defines in a natural way the computation of a discrete function by the system. Our research interest is to understand which functions are computed by which systems. In this paper we focus on identifying sets of systems that are {\em universal}, in the sense that they can collectively compute all finite functions. We are able to identify several such classes of fully synchronous systems. In particular, among other results, we prove the universality of the set of all graphs with at least one robot, of any set of graphs with at least two robots whose quotient graphs contain arbitrarily long paths, and of any set of graphs with at least three robots and arbitrarily large finite girths.

We then focus on the minimum size that a network must have for the robots to be able to compute all functions on a given finite set. We are able to approximate the minimum size of such a network up to a factor that tends to 2 as $n$ goes to infinity.

The main technique we use in our investigation is the {\em simulation} between algorithms, which in turn defines {\em domination} between systems. If a system dominates another system, then it can compute at least as many functions. The other ingredient is constituted by {\em path} and {\em ring} networks, of which we give a thorough analysis. Indeed, in terms of implicit function computations, they are revealed to be fundamental topologies with important properties. Understanding these properties enables us to extend our results to larger classes of graphs, via simulation.
\end{abstract}

\section{Introduction}\label{s1}
\noindent Consider a network, represented as a finite  graph  $G$,  where the vertices are  unlabeled, and  edge labels are possibly not  unique.
In $G$ operate $k$  {\em oblivious mobile robots} (or simply ``robots''),
that is, indistinguishable  computational entities with no memory,  located at the vertices of the network, and capable of moving
from vertex to neighboring vertex of $G$.
Robots are activated by an adversarial  {\em scheduler} $S$.
Whenever  activated,
a robot  observes the location of the other robots in the graph   (the current {\em configuration}); it computes a destination (a neighboring
vertex or the current location);
and it moves there. The computation of the destination is made by executing  an algorithm, the same for all robots,  whose sole input is the
current configuration.
The current activity terminates after the move, and no memory of the computation is retained; in other words, the entities are {\em stateless}.
The overall system is represented by the  triplet $(G,k,S)$. Notice that, even if the algorithm $A$ the robots execute is deterministic,
its executions may still be non-deterministic. Indeed, since
the network's port numbers may not be unique,  it may be impossible for an algorithm to unambiguously indicate where each
 robot has to move. 
%
This model, introduced by Klasing,  Markou, and Pelc \cite{KlMP08} as an extension of the  model of oblivious robots in continuous spaces (e.g., \cite{FlPS12}),
has been extensively employed and investigated,  focusing
 on
  basic
  problems  in specific
  classes of graphs under different schedulers:  {\em gathering} and {\em scattering} (e.g., 
  \cite{DaDKN16,DaDN14a,DaDN14b,DaDNNS15,ElB11,GuP13,IzIKO10,KaLOT12,KlKN10,KlMP08,MiPST14,OoT15}),
  and  {\em exploration} and {\em traversal}
  (e.g., \cite{BlMPT10,BoMPT11,ChFMS10,DeLPT14,DePT13,FIPS10,FlIPS11,FlIPS13,KoN09,LaGT10}).
  Note that, with the exception of \cite{ChFMS10}, the literature assumes unlabelled edges.
  In this paper, we consider both labelled and unlabelled edges, and focus on the fully synchronous scheduler ${\cal F}$, which simply activates every robot at every turn. We then
 embark on a different, more general, type of investigation.

%

Consider the system $(G,k,{\cal F})$.
Whenever the robots move in the graph according to algorithm $A$,   the system transitions  from the current
configuration to a (possibly) different one.
The obliviousness of the robots implies that always the same (or equivalent) transition occurs   from the same
given configuration.
Consider now the  {\em configuration graph} where there is a directed edge from one configuration to another if
some  algorithm  dictates such a transition.
Then the  execution of  $A$  in  $(G,k,{\cal F})$ from a given configuration is just a walk in this graph from that
configuration. The execution can be viewed in a natural way as the computation of a discrete function $f$ by the system, where $f$ maps a configuration $C$ into the configuration $f(C)$ reached by executing (one step of) $A$ from $C$, defining a subgraph of the configuration graph, called function graph.
 The concept of function computation and  function graph are formally defined in Section~\ref{s2}.

We seek to understand which  functions are computed by which systems. Knowing the structure of such functions gives us information on the robots' behavior as they execute an algorithm, and what tasks the robots can and cannot perform in a network. For instance, if an algorithm computes a function whose graph has no cycles, it means that the robots will eventually be stationary regardless of their initial position; if the function has a unique fixed point, it means that the algorithm solves a \emph{pattern formation} problem. On the other hand, if the function's graph has only cycles of length $p>1$ (possibly with some ``branches'' attached), the robots are collectively implementing a \emph{self-stabilizing clock} of period $p$. If such graphs can be embedded in the configuration graph, then we know that such algorithms exist, and that the corresponding problems are solvable in the system.

In this paper we focus  on identifying sets of systems that are {\em universal}, in the sense that they compute all finite functions. In Section~\ref{s4}, we identify several  classes of universal  fully synchronous  systems.
In particular, among other
results, we prove that

\begin{theorem*}
The following families of systems are universal:
\begin{compactitem}
\item[(a)] $\left\{(G,1,{\cal F})\mid \mbox{$G$ is an unlabeled network}\right\}$,
\item[(b)] $\left\{(G_n,2,{\cal F})\mid \mbox{the quotient graph of $G_n$ contains a sub-path of length at least $n$}\right\}$,
\item[(c)] $\left\{(G_n,3,{\cal F})\mid \mbox{the girth of $G_n$ is at least $n$ and finite}\right\}$.
\end{compactitem}
\end{theorem*}

In Section~\ref{s5}, we focus on computing  discrete functions using the smallest possible networks,  perhaps at the cost of employing a large numbers of robots.
In particular, for a given finite set $X$, we study the minimum size that a network must have for the robots to be able
to compute all functions from $X$ to $X$.
We are able to approximate the minimum size of such a network
up to a factor that tends to 2 as $n$ goes to infinity.

The main tool we use in our investigation  is the  {\em simulation} between algorithms, which
in turn defines {\em domination}  between systems.
If  system $\Psi$ dominates system $\Psi'$, then  $\Psi$ computes at least all the functions computed by $\Psi'$.
The other tool is constituted by the {\em path} and {\em ring} graphs (Section~\ref{s3}). These are the main ingredients of all our stronger results, because rings and paths are fundamental topologies with important properties that can be extended to other graphs via simulation.

\section{Definitions}\label{s2}

\noindent In this section we introduce the models of mobile robots that we are going to study. Informally, we consider networks with port numbers, which are represented as graphs where each vertex has a label on each outgoing edge. Port numbers are not required to be unique, which allows us to model anonymous networks with unlabeled edges, as well.

On a network we may place any number of robots, which are indistinguishable mobile entities with no memory. At all times, each robot must be located at a vertex of the network, and any number of robots may occupy the same vertex. All robots follow the same algorithm, which takes as input the network and the robots' positions, and tells each robot to which adjacent vertex it has to move next (or it may tell it to stay still). Time is discretized, and we assume that robots can move to adjacent vertices instantaneously.

Even if algorithms are deterministic, their executions may still be non-deterministic. This is partly because the network's port numbers may not be unique, and therefore it may be impossible for an algorithm to unambiguously indicate where each robot has to move. Another potential source of non-determinism is the scheduler, which is an adversary that decides which robots are going to be activated next. In this paper we will focus on the fully synchronous scheduler, which simply activates every robot at every turn. We will also briefly discuss the semi-synchronous scheduler in Section~\ref{s6}.

\paragraph{Labeled Graphs.}
A \emph{labeled graph} is a triplet $G=(V,E,\ell)$, where $(V,E)$ is an undirected graph called the \emph{base graph}, and $\ell$ is a function that maps each ordered pair $(u,v)$, such that $\{u,v\}\in E$, to a non-negative integer called \emph{label}. A labeled graph is also referred to as a \emph{network}. A network is \emph{unlabeled} if all its labels are equal.

An \emph{automorphism} of a labeled graph $G=(V,E,\ell)$ is a permutation $\alpha$ of $V$ preserving adjacencies and labels, i.e., for all $u,v\in V$, if $\{u,v\}\in E$, then $\{\alpha(u),\alpha(v)\}\in E$ and $\ell(u,v)=\ell(\alpha(u),\alpha(v))$. If there exists an automorphism that maps a vertex $u$ to a vertex $v$, then $u$ and $v$ are \emph{equivalent vertices} in $G$. The \emph{quotient graph} $G^\ast$ is the labeled graph $G$ obtained by identifying equivalent vertices, and preserving adjacencies and labels.

\paragraph{Configuration Spaces.}
Let $\mathbb N_n=\{0,1,\cdots,n-1\}$, for every $n\geq 1$. An \emph{arrangement} of $k$ robots on a network $G=(V,E,\ell)$ is a mapping from $\mathbb N_k$ to $V$. An arrangement specifies the locations of $k$ \emph{distinguishable} robots on a network whose vertices are all \emph{distinguishable}. However, we ultimately intend to model \emph{identical} robots, which \emph{cannot distinguish} between equivalent vertices of the network, unless such vertices are occupied by different amounts of robots. The following definition serves this purpose: two arrangements $a_1, a_2\colon \mathbb N_k \to V$, are \emph{equivalent} if there exist an automorphism $\alpha\colon V\to V$ and a permutation $\pi$ of $\mathbb N_k$ such that $\alpha\circ a_1 = a_2\circ \pi$.

The \emph{configuration space} $\mathcal C(G,k)$, where $G$ is a network and $k$ is a positive integer, is the quotient of the set of arrangements of $k$ robots on $G$ under the above equivalence relation between arrangements. The elements of the configuration space are called \emph{configurations}.

Say that an arrangement $a$ is equivalent to itself under an automorphism $\alpha$ and a permutation $\pi$, as defined above. Then, whenever $\alpha(v)=v'$ and $\pi(r)=r'$, we say that $v$ and $v'$ are \emph{equivalent vertices} in $a$, and $r$ and $r'$ are \emph{equivalent robots} in $a$. 

A class of \emph{indistinguishable} vertices $U$ (respectively, a class of \emph{indistinguishable} robots $R$) of a configuration $C\in\mathcal C(G,k)$ is a mapping from each arrangement $a\in C$ to an equivalence class of vertices $U_a$ (respectively, an equivalence class of robots $R_a$) of $a$ such that, for all $a_1,a_2\in C$ and all automorphisms $\alpha$ and permutations $\pi$ under which $a_1$ and $a_2$ are equivalent, $\alpha(U_{a_1})=U_{a_2}$ (respectively, $\pi(R_{a_1})=R_{a_2}$).

\paragraph{Configuration Graphs.}
While the configuration space contains all the configurations that are distinguishable, either by the base graph's topology, or by the labels, or by the robots' positions, the \emph{configuration graph} specifies which configurations can reach which other configurations ``in one step''. Of course, this depends on a notion of algorithm, and on a notion of scheduler.

An \emph{algorithm} for $k$ robots on a network $G$ is a function that maps a pair $(C,U)$ into a set $U'$, where $C\in\mathcal C(G,k)$ (describing the network's configuration at the moment the algorithm is executed), and $U$ and $U'$ are classes of indistinguishable vertices of $C$ (indicating the executing robot's location and its destination, respectively) such that, for every arrangement $a\in C$ and every vertex $u\in U(a)$, there exists a vertex $u'\in U'(a)$ such that either $u=u'$ or $u'$ is adjacent to $u$. According to this definition, a robot can only specify its destination as a class of indistinguishable vertices, representing either a null movement or a movement to some adjacent vertex.

An \emph{execution} for $k$ robots in a network $G$ is a sequence of configurations of $\mathcal C(G,k)$. A \emph{scheduler} for $k$ robots in a network $G$ is a binary relation between algorithms and executions. The \emph{possible} executions of an algorithm under some scheduler are the executions that correspond to the algorithm under the relation specified by such a scheduler. A \emph{system of oblivious mobile robots} is a triplet $\Psi=(G,k,S)$, where $G$ is a labeled graph, $k\geq 1$, and $S$ is a scheduler for $k$ robots in $G$.

The \emph{configuration graph} $\mathcal G(\Psi)=(\mathcal C(G,k),\mathcal E(\Psi))$, where $\Psi=(G,k,S)$ is a system of oblivious mobile robots, is a directed graph on the configuration space $\mathcal C(G,k)$, where $(C,C')\in\mathcal E(G,k)$ if there is an algorithm $A$ and a possible execution $E=(C_i)_{i\geq 0}$ of $A$ under $S$, such that there exists an index $i$ satisfying $C=C_i$ and $C'=C_{i+1}$.

The \emph{deterministic configuration graph} $\mathcal G'(\Psi)=(\mathcal C(G,k),\mathcal E'(\Psi))$, where $\Psi=(G,k,S)$ is a system of oblivious mobile robots, is a directed graph on the configuration space $\mathcal C(G,k)$, where $(C,C')\in\mathcal E'(G,k)$ if there is an algorithm $A$ such that, for all possible executions $E=(C_i)_{i\geq 0}$ of $A$ under $S$, and for every index $i$ satisfying $C=C_i$, we have $C'=C_{i+1}$.

Intuitively, $\mathcal G'(\Psi)$ is a subgraph of $\mathcal G(\Psi)$ whose edges represent moves that can be deterministically done by the robots, i.e., on which all the scheduler's choices yield the same result. If $\mathcal G(\Psi)=\mathcal G'(\Psi)$, then $\Psi$ is said to be a \emph{deterministic} system.

\paragraph{Fully Synchronous Scheduler.} Given an algorithm $A$ for $k$ robots on a network, we say that a configuration $C'$ \emph{yields} from a configuration $C$ under algorithm $A$ if, for every arrangement $a\in C$ there is an arrangement $a'\in C'$ such that, for every $r\in \mathbb N_k$, either $a(r)=a'(r)$ or $a(r)$ is adjacent to $a'(r)$ and, if $U$ is the class of indistinguishable vertices of $C$ such that $a(r)\in U(a)$, then $a'(r)\in U'(a)$, where $U'=A(C,U)$. The \emph{fully synchronous scheduler} $\mathcal F$ is defined as follows: $(A,E=(C_i)_{i\geq 0})\in \mathcal F$ if, for every $i\geq 0$, $C_{i+1}$ yields from $C_i$. In the rest of the paper, we will write $\FS{G,k}$ instead of $(G,k,\mathcal F)$.

In other words, the fully synchronous scheduler lets every robot move at every turn to the destination it computes. However, if a robot's destination consists of several indistinguishable vertices, the scheduler may arbitrarily decide to move the robot to any of those vertices, provided that it can be reached in at most one hop. All these choices are made by the scheduler at each turn and for each robot, independently.

\paragraph{Simulating Algorithms.} To define the concept of simulation, we preliminarily define a relation on executions. Given an execution $E=(C_i)_{i\geq 0}$ for $k$ robots on a network $G$, an execution $E'=(C'_i)_{i\geq 0}$ for $k'$ robots on a network $G'$, and a surjective partial function $\varphi\colon \mathcal C(G,k)\to \mathcal C(G',k')$, we say that $E$ is \emph{compliant} with $E'$ under $\varphi$ if either $\varphi$ is undefined on $C_0$, or there exists a weakly increasing surjective function $\sigma\colon\mathbb N\to\mathbb N$ such that, for every $i\in \mathbb N$, $\varphi$ is defined on $C_i$, and $\varphi(C_i)=C'_{\sigma(i)}$.

An algorithm $A$ under system $\Psi$ \emph{simulates} an algorithm $A'$ under system $\Psi'$ if there is a surjective partial function $\varphi\colon \mathcal C(G,k)\to \mathcal C(G',k')$ such that each execution of $A$ under $\Psi$ is compliant under $\varphi$ with at least one execution of $A'$ under $\Psi'$.

In this definition, $\varphi$ ``interprets'' some configurations of the simulating system $\Psi$ as configurations of the simulated system $\Psi'$, in such a way that every configuration of $\Psi'$ is represented by at least one configuration of $\Psi$. Moreover, the definition of compliance ensures that the simulating algorithm $A$ makes configurations transition in a way that agrees with $A'$ under $\varphi$.

\paragraph{Computing functions.} We define the implicit computation of a function as the simulation of a system consisting in a single robot on a network whose shape is given by the function itself.

The network \emph{induced} by a function $f\colon X\to X$ is defined as $\Gamma_f=(X,f,\ell)$, where $\ell: (u,v)\mapsto v$. Hence the base graph of $\Gamma_f$ has edges of the form $(x,f(x))$, and the labeling $\ell$ makes all vertices of $\Gamma_f$ distinguishable from each other. The algorithm $A_f$ \emph{associated} to the function $f$ is the algorithm for one robot on $\Gamma_f$ that always makes the robot move from any vertex $x\in X$ to the vertex $f(x)$.

We say that an algorithm $A$ \emph{computes} a function $f\colon X\to X$ under system $\Psi$ if it simulates the algorithm $A_f$ under $\FS{\Gamma_f,1}$.

What this definition intuitively means is that each element of $X$ is represented by a set of robot configurations; an algorithm computes $f$ if any execution from a configuration representing $x\in X$ eventually yields a configuration representing $f(x)$ without passing through configurations that represent other elements of $X$ (or that represent no element of $X$).


If an algorithm under system $\Psi$ computes a function $f$ (respectively, simulates an algorithm $A'$), then we say that $\Psi$ computes $f$ (respectively, simulates $A'$). Moreover, a system $\Psi$ \emph{dominates} $\Psi'$ if every algorithm under $\Psi'$ is simulated by some algorithm under $\Psi$.

We use the notation $X\preceq Y$ to indicate all the concepts defined above: $X$ may be a function computed by an algorithm $Y$ (under some system), or it can be an algorithm simulated by $Y$, or a system dominated by a system $Y$, etc.

%
%

\section{Basic Results}\label{s3}

\begin{proposition}
The relation $\preceq$ is transitive.
\end{proposition}
\begin{proof}
Let us prove that, if $A''\preceq A'$ and $A'\preceq A$, then $A''\preceq A$, where $A$ is an algorithm under system $\Psi=(G,k,S)$, $A'$ is an algorithm under system $\Psi'=(G',k',S')$, and $A''$ is an algorithm under system $\Psi''=(G'',k'',S'')$. All other cases (i.e., those involving functions or systems instead of only algorithms) trivially follow from this basic case.

Let the surjective partial function $\varphi\colon\mathcal C(G,k)\to \mathcal C(G',k')$ express the fact that $A$ simulates $A'$, and let the surjective partial function $\varphi'\colon\mathcal C(G',k')\to \mathcal C(G'',k'')$ express the fact that $A'$ simulates $A''$. Let us show that $\mu=\varphi'\circ\varphi$, which is a surjective partial function from $\mathcal C(G,k)$ to $\mathcal C(G'',k'')$, expresses the fact that $A$ simulates $A''$.

Let $E=(C_i)_{i\geq 0}$ be an execution of $A$ under $\Psi$. If $\mu$ is undefined on $C_0$, then $E$ complies with any execution of $A''$, by definition. So, let us assume that $\mu$ is defined on $C_0$, and therefore $\varphi$ is defined on $C_0$, and $\varphi'$ is defined on $\varphi(C_0)$. This implies that, for all $i\geq 0$, $\varphi$ is defined on $C_i$, by definition of simulation. Moreover, $E$ is compliant under $\varphi$ with some execution $E'=(C'_i)_{i\geq 0}$ of $A'$ under $\Psi'$. So there is a weakly increasing surjective function $\sigma\colon\mathbb N\to\mathbb N$ such that $\varphi(C_i)=C'_{\sigma(i)}$.

In particular, $\varphi(C_0)=C'_{\sigma(0)}=C'_0$. But we know that $\varphi'$ is defined on $\varphi(C_0)$, and hence on $C'_0$. By definition of computation, $\varphi'$ is therefore defined on each $C'_i$, with $i\geq 0$. Moreover, $E'$ is compliant under $\varphi'$ with some execution $E''=(C''_i)_{i\geq 0}$ of $A''$ under $\Psi''$. So there is a weakly increasing surjective function $\sigma'\colon\mathbb N\to\mathbb N$ such that $\varphi'(C'_i)=C''_{\sigma'(i)}$.

Now, let $i\in\mathbb N$ be an index. By the above, $\mu$ is defined on $C_i$, and
$$\mu(C_i)=\varphi'(\varphi(C_i))=\varphi'(C'_{\sigma(i)})=C''_{\sigma'(\sigma(i))}.$$
Since $\sigma'\circ\sigma\colon\mathbb N\to\mathbb N$ is itself a weakly increasing surjective function, this implies that $E$ is compiant with $E''$ under $\mu$, and hence $A$ simulates $A''$.
\end{proof}

\begin{corollary}
If a system $\Psi$ dominates a system $\Psi'$, then all functions computed by $\Psi'$ are also computed by $\Psi$.
\end{corollary}
\begin{proof}
Suppose that $\Psi'\preceq\Psi$. Then, for any function $f$ such that $f\preceq \Psi'$, the transitivity of $\preceq$ implies that $f\preceq \Psi$.
\end{proof}

\subsection{General Graphs}\label{ss10}

\begin{proposition}\label{o1}
For every network $G$, the system $\FS{G,1}$ is deterministic, and its configuration graph is isomorphic to the graph obtained from the quotient graph $G^\ast$ by replacing each unoriented edge $\{u,v\}$ with the two oriented edges $(u,v)$ and $(v,u)$, and adding a self-loop $(v,v)$ to each vertex $v$.
\end{proposition}
\begin{proof}
If there is just one robot in the network, all non-deterministic choices of a neighboring vertex among equivalent ones yield equivalent configurations. This proves that $\FS{G,1}$ is deterministic.

Note that the arrangements of one robot on $G$ consist of all the mappings of the form $0\mapsto v$, with $v$ a vertex of $G$. For each such $v$, let $C_v$ be the configuration corresponding to the arrangement $0\mapsto v$, and let $[v]$ be the equivalence class of $G^\ast$ such that $v\in [v]$. By definition of equivalence, for any two arrangements $0\mapsto u$ and $0\mapsto v$, $C_u=C_v$ if and only if $v$ and $v'$ are equivalent in $G$. This implies that the relation $\varphi$ that maps $C_v$ to $[v]$ for every vertex $v$ is a well-defined bijection between the configuration space $\mathcal C(G,1)$ and the vertex set of $G^\ast$.

Let us prove that $\varphi$ induces a graph isomorphism, if the edges of $G^\ast$ are modified as in the observation's statement. Let $(C_u,C_v)$ be an edge of the configuration graph of $\FS{G,1}$. By definition of fully synchronous scheduler, this holds if and only if there is an algorithm $A_{uv}$ that takes $C_u$ and the class of indistinguishable vertices of $C_u$ containing $u$, and maps it to the class of indistinguishable vertices of $C_u$ containing $v$. Such an algorithm exists if and only if there is a vertex $u'$ indistinguishable from $u$ in $C_u$ and a vertex $v'$ indistinguishable from $v$ in $C_u$ such that either $u'=v'$ or $u'$ and $v'$ are neighbors in $G$. However, since there is a unique robot in the network, $u=u'$, and $v$ and $v'$ have the same distance from $u$. In other words, either $[u]=[v]$, or $u$ and $v$ are non-equivalent neighbors in $G$. In the first case, the edge $(C_u,C_v)$ is correctly mapped into the self-loop on $[u]$ in the modified $G^\ast$. The second case is equivalent to $[u]$ and $[v]$ being adjacent in $G^\ast$, which is true if and only if there is the directed edge $([u],[v])$ in the modified $G^\ast$.
\end{proof}

A fundamental question is whether adding robots to a network allows to compute more functions. We can at least prove that adding robots does not reduce the set of computable functions, provided that the network is not pathologically small.

\begin{theorem}
For all networks $G$ with at least three vertices and all $k\geq 1$, $\FS{G,k+1}\not\preceq \FS{G,k}$.
\end{theorem}
\begin{proof}
It suffices to show that $|\mathcal C(G,k+1)|>|\mathcal C(G,k)|$. For each configuration in $\mathcal C(G,k)$, choose a vertex that contains the largest number of robots, and add one robot to it. This way we obtain $|\mathcal C(G,k)|$ distinct configurations of $\mathcal C(G,k+1)$. We can generate yet another configuration by placing $\lfloor (k+1)/2\rfloor$ robots on a vertex, $\lfloor (k+1)/2\rfloor$ robots on another vertex, and the remainder on a third vertex.
\end{proof}

We can also show that a single robot does not compute more functions than $k\geq 1$ robots, in any network $G$.

\begin{theorem}\label{t:1tok}
For all networks $G$ and all $k\geq 1$, $\FS{G,1}\preceq \FS{G,k}$.
\end{theorem}
\begin{proof}
In the simulation we use only the configurations of $\FS{G,k}$ in which all robots lie in equivalent vertices of $G$. Then each robot pretends to be the only robot in the network, and makes the move that the unique robot of $\FS{G,1}$ would make. This is a well-defined simulation even if $\FS{G,k}$ is not deterministic, due to Proposition~\ref{o1}.
\end{proof}

We can extend this idea to show that $\FS{G,k}\preceq \FS{G,2k}$, provided that $\FS{G,2k}$ is deterministic.

\begin{theorem}\label{p:corspecial2}
$\FS{G,k}\preceq \FS{G,k'}$, provided that $\FS{G,k'}$ is deterministic and $k'\geq 2k$.
\end{theorem}
\begin{proof}
The configurations of $\FS{G,k'}$ that we use in our simulation are only those in which there is a (unique) vertex $v$ occupied by at least $k'-k+1$ robots. Each of these configurations is mapped to the configuration of $\FS{G,k}$ that is obtained by removing $k'-k$ robots from $v$. This mapping is surjective. The simulation can be carried out because $\FS{G,k'}$ is deterministic, and therefore robots occupying the same vertex can never be separated, implying that there is always going to be a vertex with at least $k'-k+1$ robots.
\end{proof}

Finally, we conjecture that adding robots increases a system's computational capabilities.

\begin{conjecture}\label{con:1}
For all networks $G$ and all $k\geq 1$, $\FS{G,k}\preceq \FS{G,k+1}$.
\end{conjecture}

\subsection{Path Graphs}\label{ss11}

\noindent The first special type of network we consider is the one whose base graph consists of a single \emph{path}. This fundamental configuration will turn out to be of great importance in Sections~\ref{s4} and~\ref{s5}, when studying universal classes of systems. In terms of labeling, we focus on two extreme cases: a labeling that gives a consistent orientation to the whole network (i.e., each vertex in the path has port labels indicating which neighbor is on the ``left'' and which one is on the ``right''), and the anonymous unlabeled network. In the first case we have an \emph{oriented path}, and in the second case we have an \emph{unoriented path}. By $\OP nk$ and $\UP nk$ we denote, respectively, the oriented and the unoriented path with $n$ vertices and $k$ robots, under the fully synchronous scheduler.

\paragraph{Oriented Paths.} Let us study the configuration graph of $\OP nk$. Since the path has an orientation, no two vertices are equivalent. Therefore, by Proposition~\ref{o1}, $\mathcal G\left(\OP n1\right)$ consists of a path of length $n$ with bidirectional edges and a self-loop on each vertex. In general, the configuration space of $\OP nk$ is in bijection with the set of weakly increasing $k$-tuples of integers in $\mathbb N_n$. Hence, for a fixed $k$, the size of the configuration space is
$$\binom{n+k-1}{k}\ \sim\ \frac{n^k}{k!}.$$

If these $k$-tuples are thought of as points of $\mathbb R^k$, they constitute the set of lattice points in the $k$-dimensional simplex whose $k+1$ vertices have the form $(0,0,\cdots,0,1,1,\cdots,1)$. This simplex has $k+1$ facets, two of which correspond to configurations in which the first or the last vertex of the network is occupied by a robot, while the other $k-1$ facets correspond to configurations in which exactly $k-1$ vertices are occupied (i.e., exactly two robots share a vertex).

The edges of the configuration graph (that are not self-loops) connect bidirectionally all pairs of points whose Chebyshev distance is at most 1, with the exception of the points that lie on the aforementioned $k-1$ facets. Indeed, since no algorithm can separate two robots that occupy the same vertex, it follows that those facets (as well as all their intersections) can never be left once they are reached. Figure~\ref{f1}(a) shows the configuration graph of $\OP 52$.

\begin{property}\label{p:grid}
For all $n\geq 1$, $\mathcal G\left(\OP{2n}2\right)$ contains an $n\times n$ grid with bidirectional edges and self-loops.
\end{property}

Since an oriented path gives the robots a sense of direction, an algorithm can unambiguously indicate to which neighbor each robot is supposed to move. Therefore $\mathcal G\left(\OP nk\right)=\mathcal G'\left(\OP nk\right)$.

\begin{property}\label{p:opdet}
For all $n,k\geq 1$, the system $\OP nk$ is deterministic.
\end{property}

\paragraph{Unoriented Paths.} Let us study the configuration graph of $\UP nk$. Since the network in this system is unlabeled, if two vertices are symmetric with respect to the center of the path, they are equivalent. So, the configuration space is in bijection with the set of weakly increasing $k$-tuples of integers in $\mathbb N_n$, where each $k$-tuple $(a_1,\cdots,a_k)$ is identified with its ``symmetric'' one, $(n-a_k-1,\cdots,n-a_1-1)$. Elementary computations reveal that, if $k$ is fixed, these $k$-tuples are
$$\frac 12\cdot \binom{n+k-1}{k} + O\left(n^{\lfloor k/2\rfloor}\right)\ \sim\ \frac{n^k}{2\cdot k!}.$$

Geometrically, the configuration space of $\UP nk$ can be represented as the set of lattice points in a truncated $k$-dimensional simplex, which is obtained by cutting the simplex of $\OP nk$ roughly in half, along a suitable hyperplane. Figures~\ref{f1}(b) and~\ref{f1}(c) show the configuration graphs of $\UP 52$ and $\UP 62$.

If $n$ is even, we have $\mathcal G\left(\UP nk\right)=\mathcal G'\left(\UP nk\right)$, because each robot has a unique closest endpoint of the path, which it can use to specify unambiguously in which direction it intends to move. However, if $n>1$ is odd and $k\geq 2$, the two graphs differ. Indeed, if the configuration is symmetric and the central vertex is occupied by more than one robot, then it is impossible to guarantee that all the central robots will move in the same direction: the adversary will decide how many of these robots go left, and how many go right. For instance, $\mathcal G'\left(\UP 52\right)$ differs from $\mathcal G\left(\UP 52\right)$ in that the vertex in $(2,2)$ has no outgoing edges in $\mathcal G'\left(\UP 52\right)$, because these correspond to non-deterministic moves.

\begin{property}\label{p:updet}
For all $n,k\geq 1$, the system $\UP{2n}k$ is deterministic.
\end{property}

\begin{figure}[!htb]
\centering
\begin{minipage}{0.495\textwidth}
\centering
\begin{tabular}{c@{\qquad}c@{\qquad}c}
\includegraphics[scale=0.75]{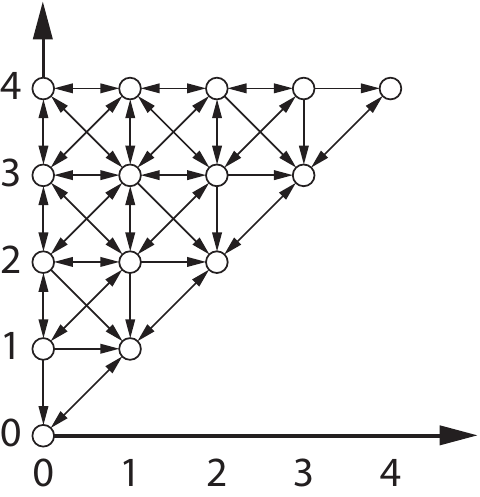} &
\includegraphics[scale=0.75]{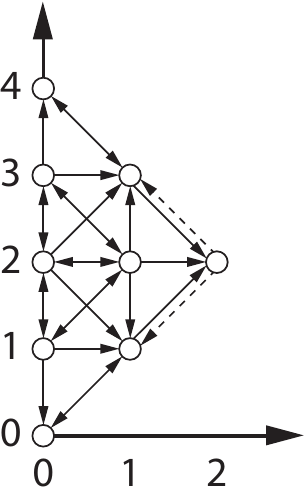} &
\includegraphics[scale=0.75]{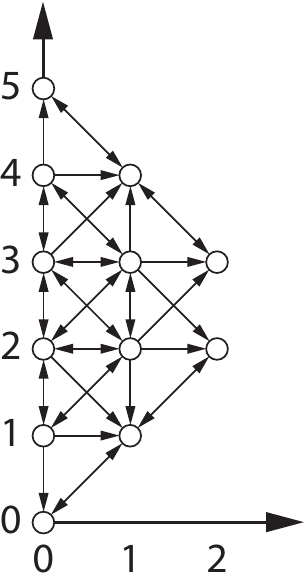} \\
(a) $\mathcal G\left(\OP 52\right)$ & (b) $\mathcal G\left(\UP 52\right)$ & (c) $\mathcal G\left(\UP 62\right)$
\end{tabular}
\end{minipage}
\begin{minipage}{0.495\textwidth}
\vspace{0.2cm}
\centering
\begin{tabular}{c@{\qquad\qquad\quad}c}
& \includegraphics[scale=0.75]{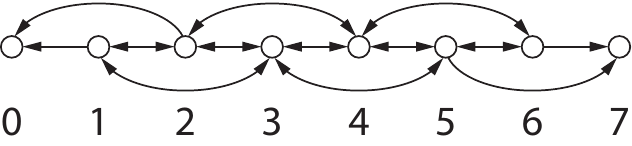} \\
& (d) $\mathcal G\left(\OR{14}2\right)$ \\
& \\
& \includegraphics[scale=0.75]{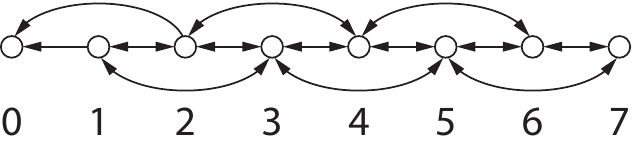} \\
& (e) $\mathcal G\left(\OR{15}2\right)$
\end{tabular}
\end{minipage}
\begin{tabular}{c@{\quad}c@{\quad}c}
& & \\
\includegraphics[scale=0.75]{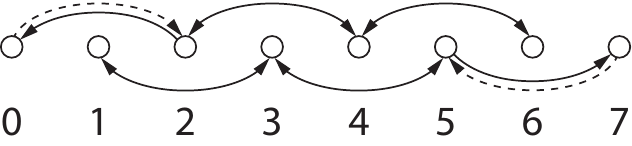} &
\includegraphics[scale=0.75]{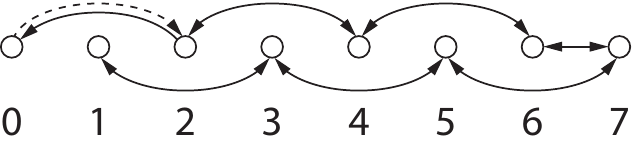} &
\includegraphics[scale=0.75]{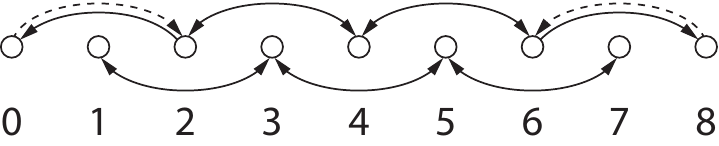} \\
(f) $\mathcal G\left(\UR{14}2\right)$ & (g) $\mathcal G\left(\UR{15}2\right)$ & (h) $\mathcal G\left(\UR{16}2\right)$
\end{tabular}
\caption{Configuration graphs of some oriented and unoriented paths and rings. Dashed arrows represent non-deterministic moves. For clarity, self-loops have been omitted from all vertices.}
\label{f1}
\end{figure}

\subsection{Ring Graphs}\label{ss12}

\noindent Now we consider \emph{ring} networks, which are networks whose base graph is a single cycle. This is another fundamental class of networks, which will have a great importance in Sections~\ref{s4} and~\ref{s5}. Like a path, a ring can be \emph{oriented} if its labeling gives a consistent sense of direction to the robots in the network (i.e., each vertex has port labels indicating which neighbor lies in the ``clockwise'' direction, and which one lies in the ``counterclockwise'' direction), and \emph{unoriented} if the network is unlabeled. Therefore we have the two systems $\OR nk$ and $\UR nk$, denoting, respectively, the oriented and the unoriented ring with $n$ vertices and $k$ robots, under the fully synchronous scheduler.

\paragraph{Oriented Rings.} Let us study the structure of $\mathcal G\left(\OR nk\right)$. Note that, in a ring network, all vertices are equivalent. Therefore, by Proposition~\ref{o1}, $\mathcal G\left(\OR n1\right)$ consists of a single vertex with a self-loop. 

In the case of $k=2$ robots, a configuration is uniquely identified by the distance $d$ of the two robots on the ring, which may be any integer between $0$ and $\lfloor n/2\rfloor$. If $d=0$, the robots are bound to remain on the same vertex. If $d=n/2$ (hence $n$ is even), the robots are located on indistinguishable vertices, and they are bound to remain on indistinguishable vertices. In all other cases, it is possible to distinguish the two robots and move them independently, thanks to the orientation of the ring. Therefore, if $0<d<n/2$, an algorithm may move the two robots independently in any direction, thus adding any integer between $-2$ and $+2$ to $d$ (subject to the $0\leq d\leq\lfloor n/2\rfloor$ constraint). Figures~\ref{f1}(d) and~\ref{f1}(e) show the configuration graphs of $\OR{14}{2}$ and $\OR{15}{2}$.

In general, the configuration space of $\OR nk$ is in bijection with the set of binary \emph{necklaces} of length $n$ and density $k$, i.e., the binary strings having $k$ zeros and $n-k$ ones, taken modulo rotations. To count them, the P\'olya enumeration theorem can be applied, as in~\cite{riordan}. If $k$ is fixed, the configuration space has size
$$\frac 1{n+k}\ \cdot \!\!\!\sum_{d|\gcd(k,n)}\!\!\!\!\!\phi(d)\cdot \binom{(n+k)/d}{k/d}\ \sim\  \frac{n^{k-1}}{k!},$$
where $\phi$ is Euler's totient function. Since the ring is oriented, we have $\mathcal G\left(\OR nk\right)=\mathcal G'\left(\OR nk\right)$.

\begin{property}\label{p:ordet}
For all $n,k\geq 1$, the system $\OR nk$ is deterministic.
\end{property}

\paragraph{Unoriented Rings.} The structure of $\mathcal G\left(\UR nk\right)$ is similar to that of $\mathcal G\left(\OR nk\right)$, except that now two configurations are indistinguishable also if they are ``reflections'' of each other. The case $k=1$ is again trivial and yields a single configuration, but the case $k=2$ is more interesting. As before, a configuration of $\UR n2$ is identified by an integer $d$ with $0\leq d\leq\lfloor n/2\rfloor$, but the two robots are now indistinguishable, and therefore they must always make symmetric moves. Hence $d$ may only change by $-2$ or $+2$; the only exception is when $n$ is odd and $d=\lfloor n/2\rfloor$, which can change to $d=\lfloor n/2\rfloor-1$ (as well as to $d=\lfloor n/2\rfloor-2$), and vice versa. So, if $n$ is odd, $\mathcal G\left(\UR n2\right)$ is isomorphic to $\mathcal G\left(\OP{\lceil n/2\rceil}1\right)$. If $n$ is even, $\mathcal G\left(\UR n2\right)$ consists of two connected components: the one corresponding to even $d$'s is isomorphic to $\mathcal G\left(\OP{\lceil (n+2)/4\rceil}1\right)$, and the one corresponding to odd $d's$ is isomorphic to $\mathcal G\left(\OP{\lfloor (n+2)/4\rfloor}1\right)$. Figures~\ref{f1}(f),~\ref{f1}(g), and~\ref{f1}(h) show the configuration graphs of $\UR{14}{2}$, $\UR{15}{2}$, and $\UR{16}{2}$.

\begin{property}\label{p:urpaths}
For all $n,k\geq 1$, $\mathcal G\left(\UR nk\right)$ consists of either a single path or two disjoint paths.
\end{property}

In general, instead of representing the configurations of $\UR nk$ with necklaces as before, we use \emph{bracelets} of length $n$ and density $k$, i.e., binary strings having $k$ zeros and $n-k$ ones, taken modulo rotations and reflections. The size of the configuration space can be computed again with the P\'olya enumeration theorem, this time using the dihedral group instead of the cyclic group. For fixed $k$, its size is
$$\frac 1{2(n+k)}\ \cdot \!\!\!\sum_{d|\gcd(k,n)}\!\!\!\!\!\phi(d)\cdot \binom{(n+k)/d}{k/d}\: +\: O\left(n^{\lfloor k/2\rfloor}\right)\ \sim\ \frac{n^{k-1}}{2\cdot k!}.$$

With unoriented rings, the deterministic configuration graph is slightly different. If $d=0$ or $d=n/2$, the adversary may choose to keep $d$ unvaried, by making both robots always move in the same direction. Therefore, the configurations corresponding to $d=0$ and $d=n/2$ have no outgoing edges in $\mathcal G'\left(\UR nk\right)$. Other than that, the two graphs are the same.

\subsection{Domination Relations Between Paths and Rings}\label{ss32}

\noindent Next we describe how different systems of paths and rings dominate each other.

%
%
%
%
%

\begin{theorem}
For all $n,k\geq 1$ and $k'\geq 2k$, $\OP nk \preceq \OP n{k'}$, $\UP{2n}k \preceq \UP{2n}{k'}$, and $\OR nk \preceq \OR n{k'}$.
\end{theorem}
\begin{proof}
By Properties~\ref{p:opdet},~\ref{p:updet}, and~\ref{p:ordet}, the systems $\OP n{k'}$, $\UP{2n}{k'}$, and $\OR n{k'}$ are deterministic. Thus all relations follow from Theorem~\ref{p:corspecial2}. 
\end{proof}

\begin{theorem}\label{t:op2up}
For all $n,k\geq 1$, $\OP nk \preceq \UP{2n}k$.
\end{theorem}
\begin{proof}
We place the $k$ robots of $\UP{2n}k$ on the first $n$ vertices of the path, leaving the other half empty. This gives an implicit orientation to the path, and now the simulation is trivial.
\end{proof}

%
%
%

\begin{theorem}\label{t:op2or}
For all $n\geq 1$ and $k\geq 2$, $\OP nk \preceq \OR{2n}{k+1}$.
\end{theorem}
\begin{proof}
We use only the configurations of $\OR{2n}{k+1}$ having a vertex occupied by a single robot, followed in clockwise order by $n-1$ empty vertices. The simulation is performed by the other $k$ robots on the remaining $n$ vertices, which are always unambiguously identified because $k\geq 2$.
\end{proof}
%
%
%
%

\begin{theorem}\label{t:up2ur}
For all $n\geq 1$ and $k\geq 2$, $\UP nk \preceq \UR{3n-1}{k+1}$.
\end{theorem}
\begin{proof}
We use only the configurations of $\UR{3n-1}{k+1}$ having a vertex occupied by one robot, surrounded on both sides by sequences of $n-1$ empty vertices. The simulation is performed by the other $k$ robots on the remaining $n$ vertices, which are always unambiguously identified because $k\geq 2$.
\end{proof}

\section{Universality}\label{s4}

\noindent A system $\Psi$ is \emph{universal for $\mathbb N_n$} if it computes every function on $\mathbb N_n$. In this case, we write $\mathbb N_n\preceq \Psi$. Note that this extension of the relation $\preceq$ preserves its transitivity. A set of systems $\Upsilon$ is \emph{universal} if, for every $n\geq 1$ and every function $f\colon\mathbb N_n\to\mathbb N_n$, there is a system $\Psi\in\Upsilon$ that computes $f$.
%
%


One robot is sufficient for universality, even on unlabeled networks.

\begin{theorem}\label{t:univ1}
$\{\FS{G,1}\mid \mbox{$G$ is an unlabeled network}\}$ is universal.
\end{theorem}
\begin{proof}
Given a function $f\colon \mathbb N_n\to\mathbb N_n$, take the complete graph $K_n$ and attach $i$ dangling vertices to its $i$-th vertex, for all $0\leq i<n$. To compute $f$, instruct the robot to move from the vertex with $i$ dangling vertices to the one with $f(i)$ dangling vertices (which is always distinguishable).
\end{proof}

However, complete graphs are very demanding networks. If no vertex in the network has more than two neighbors, universality requires more robots: two for paths and oriented rings, and three for unoriented rings.
%

\begin{theorem}\label{t:app1}
$\left\{\OP n1\mid n\geq 1\right\}$, $\left\{\UP n1\mid n\geq 1\right\}$, $\left\{\OR n1\mid n\geq 1\right\}$, and $\left\{\UR n2\mid n\geq 1\right\}$ are not universal.
\end{theorem}
\begin{proof}
This is trivially true for $\left\{\OR n1\mid n\geq 1\right\}$, since $\OR n1$ has only one configuration.

Let the \emph{cycle function} $\lambda_m\colon \mathbb N_m\to\mathbb N_m$ be defined as $\{(i,i+1)\mid i\in \mathbb N_{m-1}\}\cup\{(m-1,0)\}$. Note that, if a system computes $\lambda_m$, its configuration graph must contain a cycle of length at least $m$. By Proposition~\ref{o1} and Property~\ref{p:urpaths}, the configuration graphs of $\OP n1$, $\UP n1$, and $\UR n2$ consist of either a single path or two disjoint paths. It follows that none of these systems can compute $\lambda_m$, for any $m\geq 3$.
\end{proof}
%

Next we show that two robots are indeed sufficient to compute every function on arbitrarily long oriented and unoriented paths and oriented rings, and that three robots are sufficient for unoriented rings. First we introduce a definition: an algorithm is \emph{sequential} if it never instructs two robots to move at the same time.

\begin{lemma}\label{l:draw}
For all $n\geq 1$, $\mathbb N_n \preceq \OP{2n}2$. The algorithms used in the computations are sequential.
\end{lemma}
\begin{proof}
Let $f\colon \mathbb N_n\to\mathbb N_n$, and consider the base graph $G$ of the network induced by $f$. Each connected component of $G$ consists of a single cycle (possibly degenerating in a vertex with a self-loop) to which some disjoint trees are attached. Each tree is rooted at the vertex that it shares with the cycle, and all edges in the tree are directed towards the root. In particular, $G$ is planar.

We would like to reduce the maximum degree of $G$ to 3. We do so by a series of local transformations of $G$, as shown in Figure~\ref{f4}. If $G$ has a vertex $v$ of degree greater than 3, we transform it as follows. If $(u,v)$ and $(v,w)$ are edges of $G$, we delete them, we create a new vertex $v_u$, and we add the edges $(v,v_u)$, $(u,v_u)$, and $(v_u,w)$. This transformation creates a new vertex of degree 3, and decreases by 1 the degree of $v$. If $v$ has initial indegree $d>3$, this operation is performed on $v$ exactly $d-3$ times, adding $d-3$ vertices to the graph.

\begin{figure}[!htb]
\centering
\includegraphics[scale=1]{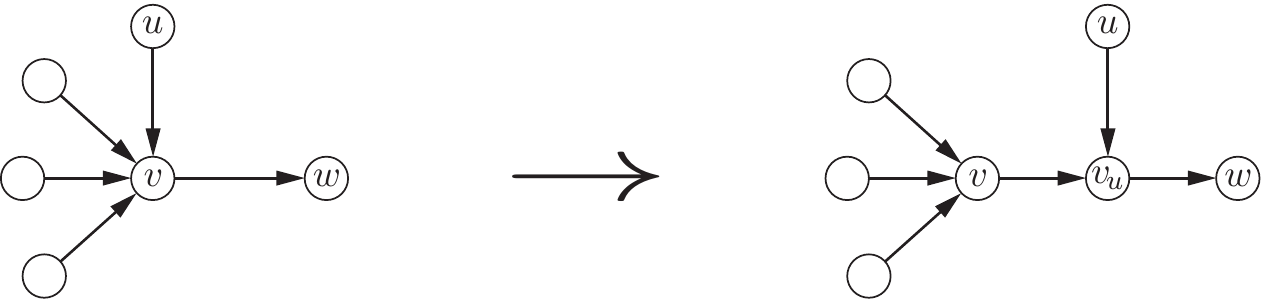}
\caption{Reducing the degree of $v$.}
\label{f4}
\end{figure}

Therefore, by repeatedly applying this transformation, we obtain a new planar graph $G'$ on at most $2n$ vertices, with maximum degree at most 3, such that $G$ is a graph minor of $G'$. As shown in~\cite{kant}, $G'$ admits a planar orthogonal grid drawing on an $n\times n$ grid. So, $G'$ is a minor of some graph $G''$ (obtained by adding vertices on some edges of $G'$) that in turn is a subgraph of an $n\times n$ grid with self-loops and bidirectional edges.

By Properties~\ref{p:grid} and~\ref{p:opdet}, this $n\times n$ grid is a subgraph of $\mathcal G'\left(\OP{2n}2\right)$. It is easy to check that $\OP{2n}2$ computes $f$ under the surjective partial function induced by the composition of these graph minor relations.

Indeed, whenever the computation reaches a configuration corresponding to a vertex of $G$ that has been splitted during the first set of transformations, it just follows the edges generated during the splitting. This is possible because all these edges represent deterministic moves, so they can be chosen by the algorithm and have to be complied by the fully synchronous scheduler.

On the other hand, if the computation reaches a configuration corresponding to an edge of $G$ that has been extended in order to embed it in the grid, the computation follows the extended edge. Again, this is possible because these edges represent deterministic moves.
\end{proof}

\begin{theorem}\label{t:pathuniv}
$\left\{\OP n2\mid n\geq 1\right\}$, $\left\{\UP n2\mid n\geq 1\right\}$, $\left\{\OR n2\mid n\geq 1\right\}$, and $\left\{\UR n3\mid n\geq 1\right\}$ are universal.
\end{theorem}
\begin{proof}
For $\left\{\OP n2\mid n\geq 1\right\}$, $\left\{\UP n2\mid n\geq 1\right\}$, and $\left\{\UR n3\mid n\geq 1\right\}$, this immediately follows from Lemma~\ref{l:draw}, Theorems~\ref{t:op2up} and~\ref{t:up2ur}, and the transitivity of $\preceq$.

For oriented rings, we need a different argument. The shape of the configuration graph of $\OR n2$ is illustrated in Figures~\ref{f1}(d) and~\ref{f1}(e). Given a function $f$ whose induced network is $\Gamma_f$, we can embed $\Gamma_f$ in $\mathcal G\left(\OR n2\right)$, provided that $n$ is large enough, for instance as shown in Figure~\ref{f3}. In general, each connected component of the base graph of $\Gamma_f$ is a cycle $L$ with some trees attached to it. for each leaf of such a tree, connected to $L$ via a path $P$, we embed a copy of $L$ in $\mathcal G\left(\OR n2\right)$ and attach a copy of $P$ to it, mapping each vertex to the appropriate vertex of $L\cup P$. The resulting surjective partial function $\varphi$ indeed defines a simulation.
\end{proof}

\begin{figure}[!htb]
\centering
\includegraphics[scale=1]{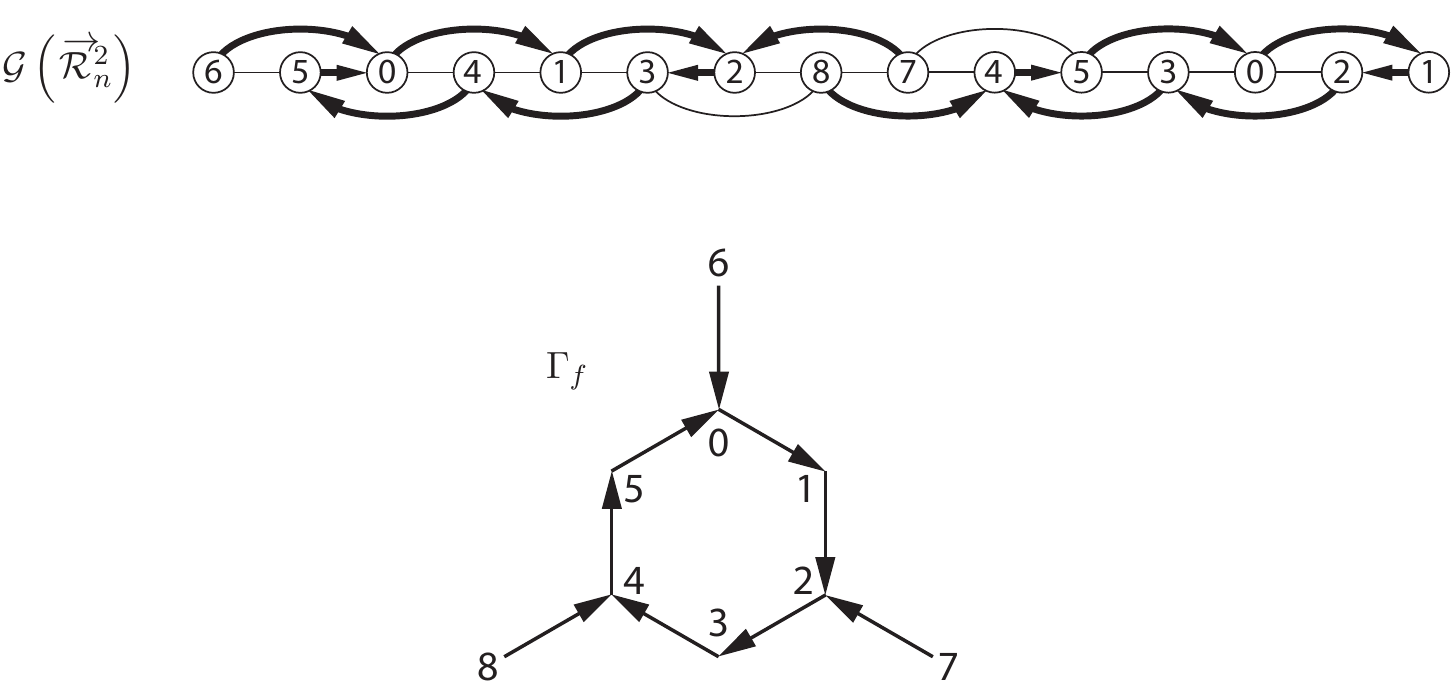}
\vspace{0.25cm}
\caption{Computing a function by two robots in an oriented ring.}
\label{f3}
\end{figure}

We can generalize this result in two directions. First, to systems of two robots on networks whose quotient graphs contain arbitrarily long paths.

\begin{theorem}\label{cor:1}
$\left\{\FS{G_n,2}\mid \mbox{$G^\ast_n$ contains a sub-path of length at least $n$}\right\}$ is universal.
\end{theorem}
\begin{proof}
We prove that $\OP n2\preceq\FS{G_n,2}$, and the universality follows from Lemma~\ref{l:draw}. By assumption, the robots can agree on an oriented path $P$ of length $n$ in $G^\ast_n$, since all its vertices are distinguishable, due to the definition of quotient graph. Then, a robot located at some vertex of $G_n$ is interpreted as lying on the corresponding vertex of $G^\ast_n$ and follows whatever algorithm it is simulating, remaining on vertices of $G_n$ corresponding to $P$; this way, the two robots can simulate $\OP n2$ on $G_n$.
%
\end{proof}

Theorem~\ref{t:pathuniv} can also be generalized to systems of three robots on networks with arbitrarily long \emph{girths} (the girth of a graph being the length of its shortest cycle, or infinity if there are no cycles).

\begin{theorem}\label{cor:2}
$\left\{\FS{G_n,3}\mid \mbox{the girth of $G_n$ is at least $n$ and finite}\right\}$ is universal.
\end{theorem}
\begin{proof}
We show that $\mathbb N_n \preceq \FS{G_{8n},3}$. By Lemma~\ref{l:draw}, it is sufficient to simulate any \emph{sequential} algorithm $A$ for $\OP{2n}2$. In particular, $A$ does not make the two robots move if they lie on the same vertex, or both would have to move, and the algorithm would not be sequential. To do the simulation, we choose a shortest cycle $L$ in $G_{8n}$, and we initially place all three robots on it: one robot $a$ will remain fixed in one vertex, then we will keep $2n-1$ empty vertices on $L$ after $a$, while the other two robots $b$ and $c$ simulate $A$ on the next $2n$ vertices, with $c$ always farthest from $a$. This makes the robots distinguishable (unless $b$ and $c$ coincide), because $L$ has length at least $8n$ (see Figure~\ref{f5}).

\begin{figure}[!htb]
\centering
\includegraphics[scale=1]{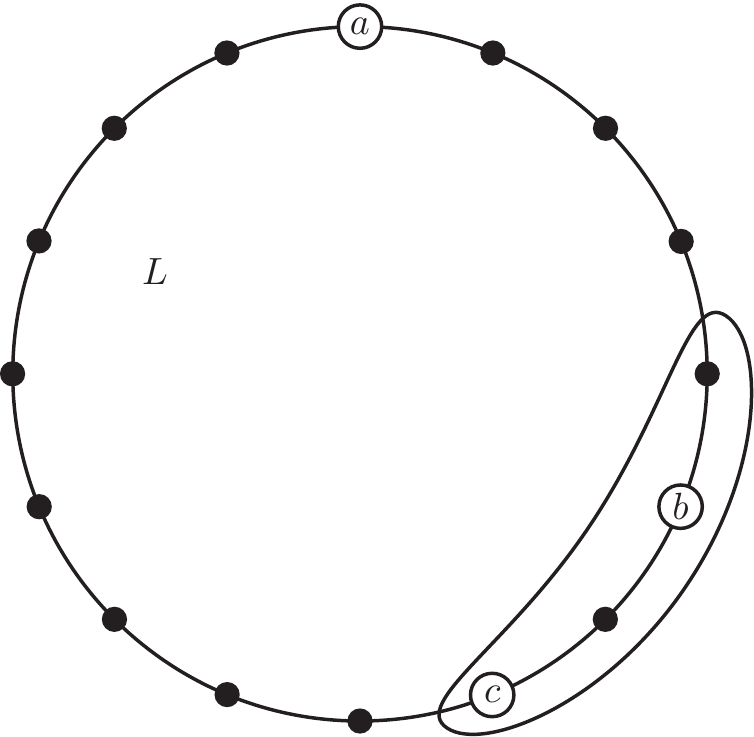}
\caption{Proof of Theorem~\ref{cor:2}.}
\label{f5}
\end{figure}

Since $L$ is a shortest cycle, there is a unique shortest path $P$ in $G_{8n}$ containing all three robots. Indeed, if this were not the case, there would be two vertices of $L$ connected by two disjoint paths of length at most $4n-1$, implying the existence of a cycle of length at most $8n-2$, and contradicting the fact that the girth is at least $8n$.

Due to the uniqueness of $P$, if $b$ and $c$ have to move within $P$, they can deterministically do so. This may cause $b$ and $c$ to switch positions, in which case we simply rename them and we pretend they have not moved.

If $c$ has to move away from $b$ (hence $b$ and $c$ are not on the same vertex, and $b$ stays still because $A$ is sequential), it proceeds along a shortest cycle that contains $P$. This could be a non-deterministic move, but it will indeed keep the robots on a same shortest cycle, although maybe not $L$. All the configurations obtained this way are mapped to the corresponding configuration of $\OP{2n}2$, and this yields a surjective partial function that describes the simulation.
\end{proof}

Our conjecture is that the above results characterize the universal classes of systems with at least three robots on unlabeled networks.

\begin{conjecture}\label{con:2}
The set $\left\{\FS{G_i,k}\mid i\geq 0\right\}$, where $k\geq 3$ and every $G_i$ is an unlabeled network, is universal if and only if either the quotient graphs $G^\ast_i$ have unboundedly long sub-paths, or the graphs $G_i$ have unboundedly long shortest cycles.
\end{conjecture}

\section{Optimizing Network Sizes}\label{s5}

\noindent In this section, our goal is to compute all the functions on $\mathbb N_n$ under the fully synchronous scheduler, using the smallest possible network, and perhaps a large number of robots. We are able to approximate the minimum size of such a network up to a factor that tends to 2 as $n$ goes to infinity, using very short oriented paths. Nonetheless, thanks to the simulation tools developed in Section~\ref{ss32}, we could as well use unoriented paths or rings, again achieving the optimum size up to factors that tend to small constants.

\begin{lemma}\label{l:ham}
For all $n,k\geq 1$, $\OP{n!}k\preceq \OP{kn}{kn(n-1)/2}$.
\end{lemma}
\begin{proof}
We divide the base graph of $\OP{kn}{kn(n-1)/2}$ into $k$ sub-paths of length $n$. In every sub-path, we place a different amount of robots on each vertex, from $0$ to $n-1$ robots. The possible placements of such robots within a sub-path correspond to the $n!$ permutations of $n$ distinct objects. It is well known that the set of permutations can be ordered in such a way that each permutation is obtained from the previous one by swapping only two adjacent objects~\cite{johnson}. If we let the $i$-th permutation under this ordering encode the $i$-th vertex of $\OP{n!}k$, we can simulate a move of $i$-th robot of $\OP{n!}k$ by simply swapping the robots occupying two adjacent vertices of the $i$-th sub-path of $\OP{kn}{kn(n-1)/2}$.
\end{proof}

\begin{theorem}\label{t:opt1}
For all $n\geq 1$, $\mathbb N_n\preceq \OP{2m}{m(m-1)}$, with $(m-1)!< 2n \leq m!$.
\end{theorem}
\begin{proof}
Immediate from Lemma~\ref{l:draw}, Lemma~\ref{l:ham} with $k=2$, and the transitivity of $\preceq$.
\end{proof}

This tells us that, on a network with $|V|$ vertices, all the functions on a set of size $2(|V|/2)!$ can be computed, provided that enough robots are available. We can also show that, on the same network, it is impossible to compute all functions on a set of size $|V|!+1$.

\begin{theorem}\label{t:opt2}
For all networks $G=(V,E,\ell)$ and all $n,k\geq 1$, if $\mathbb N_n \preceq \FS{G,k}$, then $|V|!\geq n$.
\end{theorem}
\begin{proof}
Let $A$ be an algorithm that computes the cycle function $\lambda_n=\{(i,i+1)\mid i\in \mathbb N_{n-1}\}\cup\{(n-1,0)\}$ under $\FS{G,k}$, according to the surjective partial function by $\varphi\colon \mathcal C(G,k)\to\mathbb N_n$. For any execution $E=(C_i)_{i\geq 0}$ of $A$ such that $C_0\in\varphi^{-1}(0)$, the sequence $(\varphi(C_i))_{i\geq 0}$ must span the whole range $\mathbb N_n$ infinitely often. Moreover, since $\mathcal C(G,k)$ is finite, there is a configuration $C$ that occurs infinitely many times in $E$. Hence there are two such occurrences, say $C_j=C_{j'}=C$, between which $E$ spans at least $d\geq n$ different configurations. During this fragment of the execution, no two separate sets of robots may end up on the same vertex at the same time, or they become impossible to separate deterministically, contradicting the fact that $C$ must be reached again. In particular, the number of vertices that are occupied by some robots remains the same, $q$. It follows that $d$ cannot be larger than $|V|!/(|V|-q)!$. Thus $n\leq d\leq |V|!/(|V|-q)!\leq |V|!$, as desired.
\end{proof}

%
%

\section{Further Work}\label{s6}

\noindent In addition to the fully synchronous scheduler, also a semi-synchronous and an asynchronous one can be defined. Both schedulers may activate any subset of the robots at each turn, keeping the others quiescent. The asynchronous scheduler may even delay the robots, making them move based on obsolete observations of the network.

Our universality results can be extended to both these schedulers, by observing that all the algorithms we used in our simulations are sequential. We can also prove a weaker version of Theorem~\ref{t:opt1} for these schedulers: $\mathbb N_n\preceq \OP{2m}{m}$, with $m=O(\log n)$; we have a matching lower bound for both schedulers, as well. These results indicate that the semi-synchronous and asynchronous schedulers, albeit not drastically reducing the robots' computing powers, make them somewhat less efficient.

We leave two open problems: Conjectures~\ref{con:1} and~\ref{con:2}.

\bibliographystyle{plain}

\end{document}